\documentclass[12pt]{article}

\usepackage[letterpaper,margin=1in]{geometry}
\usepackage{amsmath}
\usepackage{amssymb,amsthm}
\usepackage{enumitem}
\usepackage{needspace}
\usepackage{microtype}
\usepackage[authoryear,round]{natbib}
\usepackage{setspace}
\usepackage{xurl}
\usepackage{tikz}
\usetikzlibrary{positioning,arrows.meta}
\usepackage[colorlinks=true,allcolors=blue]{hyperref}

\allowdisplaybreaks
\renewcommand{\thesection}{\Roman{section}}

\makeatletter
\renewcommand{\@seccntformat}[1]{\csname the#1\endcsname.\quad}
\renewcommand\section{\@startsection{section}{1}{\z@}{-3.5ex \@plus -1ex \@minus -.2ex}{2.3ex \@plus.2ex}{\normalfont\Large\bfseries\raggedright}}
\renewcommand\subsection{\@startsection{subsection}{2}{\z@}{-3.25ex\@plus -1ex \@minus -.2ex}{1.5ex \@plus .2ex}{\normalfont\large\bfseries\raggedright}}
\let\old@appendix\appendix
\renewcommand{\appendix}{\old@appendix\renewcommand{\@seccntformat}[1]{Appendix~\csname the##1\endcsname.\quad}}
\newcommand{\supplementsetup}{%
  \setcounter{section}{0}\renewcommand{\thesection}{S\arabic{section}}\def\theHsection{S\arabic{section}}%
  \renewcommand{\@seccntformat}[1]{\csname the##1\endcsname.\quad}%
  \setcounter{theorem}{0}\renewcommand{\thetheorem}{S\arabic{theorem}}\def\theHtheorem{S\arabic{theorem}}%
  \setcounter{lemma}{0}\renewcommand{\thelemma}{S\arabic{lemma}}\def\theHlemma{S\arabic{lemma}}%
  \setcounter{proposition}{0}\renewcommand{\theproposition}{S\arabic{proposition}}\def\theHproposition{S\arabic{proposition}}%
  \setcounter{definition}{0}\renewcommand{\thedefinition}{S\arabic{definition}}\def\theHdefinition{S\arabic{definition}}%
  \setcounter{example}{0}\renewcommand{\theexample}{S\arabic{example}}\def\theHexample{S\arabic{example}}%
  \setcounter{remark}{0}\renewcommand{\theremark}{S\arabic{remark}}\def\theHremark{S\arabic{remark}}%
  \setcounter{equation}{0}\renewcommand{\theequation}{S\arabic{equation}}\def\theHequation{S\arabic{equation}}%
}
\makeatother
\newenvironment{figurenotes}{\par\vspace{4pt}\noindent\begin{minipage}{0.96\textwidth}\footnotesize\setstretch{1.1}\noindent }{\end{minipage}\par}

\newtheorem{theorem}{Theorem}
\newtheorem{lemma}{Lemma}
\newtheorem{proposition}{Proposition}

\theoremstyle{definition}
\newtheorem{definition}{Definition}
\newtheorem{example}{Example}
\theoremstyle{remark}
\newtheorem{remark}{Remark}

\newcommand{\Rev}{\operatorname{Rev}}
\newcommand{\Obs}{O}

\title{Protecting Losers Undermines Auction Credibility}
\author{Yuta Nakamura and Noriaki Okamoto\thanks{Nakamura: School of Economics and Business Administration, Yokohama City University, 22-2 Seto, Kanazawaku, Yokohama 236-0027, Japan (email: y\_naka@yokohama-cu.ac.jp); Okamoto: Department of Economics, Meiji Gakuin University, 1-2-37 Shirokanedai, Minatoku, Tokyo 108-8636, Japan (email: nokamoto@eco.meijigakuin.ac.jp). This work was supported by JSPS KAKENHI Grant Numbers 24K16334 (Nakamura) and 24K16338 (Okamoto).}}
\date{October 2026}

\begin{document}
\maketitle

\begin{abstract}
\noindent A second-price auction must determine its price but need not learn how far losing values lie below it. We show that protecting this information from the seller conflicts with auction credibility. In deterministic protocols with private communication, lower-tail privacy forces the first informative question to reveal whether a bidder has the highest possible value. The seller can then manipulate the winner's payment without any bidder detecting the deviation, increasing revenue at some valuation profiles and never reducing it. A descending protocol conceals losing values below the price but permits manipulation; with mandatory sale, a credible ascending protocol reveals them.
\end{abstract}


Can an auction protect losing bidders' private information when the seller cannot commit to follow its rules? A second-price auction awards the object to the highest-valued bidder and charges the second-highest value, without needing to learn lower losing values. Concealing these values can protect bidders in later auctions or negotiations \citep{RothkopfTeisbergKahn1990}. We show that under private communication, this protection gives the seller an unavoidable incentive to inflate the winner's payment.

We consider a single-object auction with values in a common finite set. The seller privately asks bidders questions, choosing each question and its recipient from earlier answers; bidders observe only their own messages, allocation, and payment. A deviation is undetectable if each bidder's observations are consistent with the rules for his own value and some values of his opponents. We use \emph{prior-free credibility}, defined by \citet[Section~5.2, Definition~20]{AkbarpourLi2020}, which rules out such deviations that raise revenue at some profile without reducing it at any other. This criterion requires neither a prior nor revenue optimality.

Figure~\ref{fig:timing} compares descending and ascending protocols implementing the second-price rule. The descending protocol asks whether values are at least a threshold and gradually lowers it. For example, suppose three bidders have values $(100,10,1)$ and bidder 1 wins in ties. The seller learns bidder 1's value of 100 before determining the price of 10, and can charge 100 without detection: a rival with value 100 explains the winner's observations, while losers receive unchanged messages and pay zero. With mandatory sale, an ascending protocol can be prior-free credible, but rising thresholds disclose low losing values: on the integer grid, bidder 3's acceptance at 1 and refusal at 2 reveal value 1, although values 1 and 2 give the same winner and price.

\begin{figure}[!t]
\centering\footnotesize
\begin{minipage}{0.96\textwidth}
\noindent\textbf{A. Descending: the winner's value is learned before the price}\\[3pt]
\begin{tikzpicture}[box/.style={draw,align=center,minimum height=14mm,text width=44mm,font=\footnotesize,inner sep=3pt},arr/.style={-{Stealth[length=2mm]}}]
\node[box] (a1) {Threshold 100\\Bidder 1 answers yes\\Seller learns $v_1=100$};
\node[box,right=5mm of a1] (a2) {Thresholds fall to 10\\Bidder 2 first answers yes\\Seller learns the price is 10};
\node[box,right=5mm of a2] (a3) {Winner 1 pays 10\\Bidder 3's values below 10\\remain indistinguishable};
\draw[arr] (a1) -- (a2); \draw[arr] (a2) -- (a3);
\end{tikzpicture}\\[3pt]
The seller can charge 100 without detection: a rival's value of 100 explains the winner's payment, and the losers' own messages and payments remain unchanged.\\[8pt]
\noindent\textbf{B. Ascending: low losing values are disclosed}\\[3pt]
\begin{tikzpicture}[box/.style={draw,align=center,minimum height=14mm,text width=44mm,font=\footnotesize,inner sep=3pt},arr/.style={-{Stealth[length=2mm]}}]
\node[box] (b1) {Threshold 2\\Bidder 3 answers no\\Seller learns $v_3=1$};
\node[box,right=5mm of b1] (b2) {Thresholds rise through 10\\Bidders 1 and 2 answer yes\\Both remain active};
\node[box,right=5mm of b2] (b3) {Threshold 11\\Bidder 2 answers no\\Bidder 1 wins and pays 10};
\draw[arr] (b1) -- (b2); \draw[arr] (b2) -- (b3);
\end{tikzpicture}\\[3pt]
If bidder 3's value were 2, his answer at 2 would be yes. The winner and payment would be unchanged, but the seller would observe a different answer.
\end{minipage}
\caption{The privacy and credibility trade-off under second pricing}
\begin{figurenotes}
Values belong to $\{0,\ldots,100\}$ and are $(100,10,1)$; ties are won by the bidder with the smallest index. Each question asks whether the recipient's value is at least the stated threshold. After acceptance at 1, a refusal at 2 identifies value 1. The rows show selected stages of the descending protocol in Section~\ref{sec:size} and the ascending protocol in Appendix~\ref{app:ascending}, which is prior-free credible under mandatory sale. All messages are private.
\end{figurenotes}
\label{fig:timing}
\end{figure}

The contextual privacy criterion of \citet{HauptHitzig2026} requires that, holding opponents' values fixed, two values yielding the same allocation and payments generate the same complete history of questions and answers. Even requiring this only for losing bidders makes the second-price rule impossible to implement (Proposition~\ref{prop:loser-contextual}). We instead impose \emph{lower-tail contextual privacy}: changes between losing values that are both strictly below the price must leave this history unchanged. Unlike a price-setting value, a lower losing value can stay fixed while other losers' values fall and reduce the same winner's payment. This weaker requirement permits distinguishing a price-setting value from a lower value even when the outcome is unchanged, but conceals distinctions among values below the price. The descending protocol satisfies it.

With at least three bidders and three possible values, no finite deterministic private-query protocol implementing the second-price rule satisfies both lower-tail contextual privacy and prior-free credibility (Theorem~\ref{thm:main}). The result permits arbitrary questions with finite answer sets, including repeated queries. Every protocol satisfying lower-tail contextual privacy admits an undetectable seller deviation that never lowers revenue and raises it at some profile.\footnote{Privacy is imposed on the prescribed protocol. Credibility tests deviations subject to individual undetectability; it does not additionally require those deviations to preserve privacy.} The same deviation therefore raises expected revenue under every distribution assigning positive probability to every valuation profile.

Privacy determines the first informative question. If two opponents had the highest possible value, any lower value of the queried bidder would yield the same outcome. Privacy therefore requires identical answers across all these lower values, so the seller cannot distinguish among them. To be informative, the question must therefore separate the highest value from all lower values. Upon learning that value, the seller can compute the bidder's subsequent truthful answers herself while querying opponents. If he wins, she sends him messages consistent with the highest payment he could owe under the rules, while losers retain their actual messages and outcomes. This never reduces revenue and raises it at some profiles; private observations prevent bidders from detecting the incompatible explanations.

Manipulation can be substantial. On fine, equally spaced grids, every privacy-preserving protocol allows an undetectable price increase close to the full range of values at some profile. In the descending protocol with independent uniform values, the expected increase converges to a positive amount as the grid becomes finer (Proposition~\ref{prop:grids}). On integer grids, credibility rules out protecting any three fixed bidders whenever they lose below the price (Theorem~\ref{thm:cost}). Thus, all but at most two bidders must each face a privacy violation somewhere, possibly at different profiles. Section~\ref{sa:general} of the Supplemental Appendix extends the impossibility to other efficient, incentive-compatible payment rules whose outcomes do not depend on lower-tail changes, on grids with at least four values.

Our paper connects two concerns highlighted by \citet{RothkopfTeisbergKahn1990}: auctioneer cheating and the future use of bidders' information. \citet{HauptHitzig2026} compare protocols by their contextual privacy violations and characterize maximally private protocols; maximality need not eliminate all unnecessary disclosure. We identify a weaker, feasible protection for losers that conflicts with credibility. \citet{MilgromSegal2020} study winner privacy in clock auctions. Under mandatory sale, our ascending construction is prior-free credible and protects the winner but reveals lower losing values. \citet{GaoTang2026} distinguish privacy from the auctioneer from privacy among bidders; ours concerns information learned by the seller.

\citet{AkbarpourLi2020} show that, in a credible auction, a bidder who can win without acting again must know his payment conditional on winning. Otherwise, the seller can substitute a higher payment consistent with his observations. We show why lower-tail privacy makes this manipulation unavoidable even when further questions are permitted: early value identification lets the seller compute the bidder's truthful answers to later questions herself. Related work studies credibility with privately informed sellers \citep{BanchioSkrzypaczYang2025} and seller-submitted shill bids \citep{KomoKominersRoughgarden2025,Shinozaki2026}. Proofs omitted from the main text appear in the Appendix.

\section{Outcomes, Information, and Credibility}\label{sec:model}

A seller allocates one object among bidders $N=\{1,\ldots,n\}$, with $n\ge 3$. Bidder $i$ has private value $v_i$ in the common finite set $V=\{a_1<\cdots<a_m\}$, where $a_1\ge 0$ and $m\ge 3$. Utility is value times allocation minus payment. Ties among bidders with the highest value are broken in favor of the bidder with the smallest index; we say that a bidder with a smaller index has \emph{priority} over one with a larger index. An efficient rule assigns the object to
\[
w(v)=\min\{i: v_i=\max_j v_j\}.
\]
Write $t_i(v)$ for bidder $i$'s payment to the seller and $F(v)=(w(v),t_1(v),\ldots,t_n(v))$ for the full outcome. We normalize loser payments to zero. Ex post individual rationality requires nonnegative truthful utility. Dominant-strategy incentive compatibility (DSIC) means that truthful reporting is weakly optimal in the induced direct rule for every fixed profile of opponents' reports.

The second-price rule charges the winner
\begin{equation}\label{eq:prices}
t_{w(v)}(v)=\ell(v):=\max_{j\ne w(v)}v_j,
\end{equation}
and nothing to losers. It is normalized, efficient, DSIC, and ex post individually rational. A bidder $j\ne w(v)$ with $v_j=\ell(v)$ is a \emph{price-setter}. Throughout, $F$ denotes this second-price rule.\footnote{On $V=\{0,\ldots,K\}$, the standard ascending clock charges the winner his smallest winning value: $\ell(v)+1$ when a price-setter has priority over him, and $\ell(v)$ otherwise. See Section~\ref{sa:clock} of the Supplemental Appendix.}

\begin{definition}[Private query protocol]\label{def:protocol}
A protocol $P$ is a finite deterministic tree. At each nonterminal node the seller privately sends one bidder a query $q$. Each query specifies a finite answer set $A_q$ and a response map $\rho_q:V\to A_q$, where $\rho_q(v_i)$ is bidder $i$'s truthful answer. The answer determines the next node. A terminal node specifies an allocation and payments. The protocol implements $F$ if truthful answers produce $F(v)$ at every profile.
\end{definition}

Questions can be arbitrary and may be redundant. For example, the query ``Is your value at least $p$?'' specifies the answer yes when $v_i\ge p$ and no otherwise. Thus, a seller who knows a bidder's value can compute his truthful answer to any specified query. Throughout, bidder $i$ answers query $q$ with $\rho_q(v_i)$, including when the seller deviates. %

A bidder observes only his ordered own query--answer sequence and his own allocation and payment, not others' messages or payments, the global order of communication, elapsed time, or a common record. A query may also state what has happened so far. Such a statement is itself a private seller message: the recipient cannot verify it with other bidders, and a deviating seller may falsify it.

Let $z^P(v)$ record the complete ordered sequence of queries and answers, including each query's recipient and any stated history, together with the final allocation and all payments. Let $\Obs^P_i(v)$ be bidder $i$'s own observation, consisting of his query--answer sequence, allocation, and payment.

\subsection{Contextual Privacy and Its Lower-Tail Restriction}

We first state the contextual privacy criterion of \citet{HauptHitzig2026} in our notation. They call the absence of any contextual privacy violation \emph{full contextual privacy}.

\begin{definition}[Contextual privacy]\label{def:contextual}
A protocol $P$ implementing $F$ is \emph{contextually private} if, for every $v\in V^n$, every bidder $i$, and every alternative value $v'_i\in V$,
\begin{equation}\label{eq:contextual}
F(v'_i,v_{-i})=F(v)\quad\Longrightarrow\quad z^P(v'_i,v_{-i})=z^P(v).
\end{equation}
It satisfies \emph{contextual privacy for losers} if \eqref{eq:contextual} is required only for bidders $i\ne w(v)$. Outcome equality ensures that such a bidder loses at both profiles.
\end{definition}

The criterion compares two values of one bidder while holding all other values fixed: if the full outcome is the same, the seller must observe the same terminal transcript. The following impossibility follows from the proof of Proposition~4 in \citet{HauptHitzig2026}, whose argument fixes the winner and forces a privacy violation for one of two losers.

\begin{proposition}[Contextual privacy for losers is impossible]\label{prop:loser-contextual}
With at least three bidders and at least two common values, no private query protocol implementing the second-price rule satisfies contextual privacy for losers.
\end{proposition}

Contextual privacy for losers requires the same transcript even when a losing value $v_i$ equals $\ell(v)$ and the alternative $v'_i$ lies below $\ell(v)$, provided both give the same winner and payment. These values nevertheless differ in whether they prevent the price from falling. If $v_i=\ell(v)$ stays fixed, the same winner $w(v)$ cannot pay less than $\ell(v)$ even if the other losers' values fall. If instead $v_i<\ell(v)$, lowering the other losers' values sufficiently reduces the payment while keeping $v_i$ and the winner unchanged. Our privacy objective is to conceal how far $v_i$ lies below $\ell(v)$, while allowing the protocol to distinguish whether $v_i$ prevents a price decrease. We therefore protect comparisons between two losing values $v_i$ and $v'_i$ both strictly below $\ell(v)$. Call bidder $i\ne w(v)$ a \emph{low loser} at $v$ if $v_i<\ell(v)$.

\begin{definition}[Lower-tail contextual privacy]\label{def:privacy}
A protocol $P$ satisfies \emph{lower-tail contextual privacy} if, for every $v\in V^n$, every loser $i\ne w(v)$, and every alternative value $v'_i\in V$,
\begin{equation}\label{eq:privacy}
\max\{v_i,v'_i\}<\ell(v)\quad\Longrightarrow\quad z^P(v'_i,v_{-i})=z^P(v).
\end{equation}
\end{definition}

For a loser $i$, $\max\{v_i,v'_i\}<\ell(v)$ implies $F(v'_i,v_{-i})=F(v)$: changing his value within this range leaves the winner and all payments unchanged. Haupt--Hitzig's criterion protects every pair of values that gives the same outcome; ours protects only pairs of losing values both strictly below the price, and is thus weaker even than contextual privacy for losers. A losing value equal to the price may be distinguished from a lower value even when the outcome is unchanged. The following example illustrates the protection and its interpretation in terms of price reductions.

\begin{example}[Two values below the price]\label{ex:price-reduction}
Take three bidders and $V=\{0,1,2,3,4\}$. Consider bidder 3's values of 1 and 2 at $(4,3,1)$ and $(4,3,2)$. At both profiles, bidder 1 wins and pays 3. Lowering bidder 2's value to 1 gives $(4,1,1)$ and $(4,1,2)$, respectively. This reduces bidder 1's payment while keeping bidder 3's value fixed: the new payments are 1 and 2, respectively. By contrast, if bidder 3's value is 3, giving $(4,3,3)$, bidder 1 cannot pay less than 3 as long as bidder 3 remains at 3. Lower-tail contextual privacy requires the same transcript at $(4,3,1)$ and $(4,3,2)$ but permits either profile to be distinguished from $(4,3,3)$.
\end{example}

The example suggests an equivalent formulation. Consider two values of a loser that give the same outcome, holding everyone else's values fixed. For each value, consider whether the same winner could pay less at some profile that keeps this loser's value fixed but changes other bidders' values. Require the same transcript whenever such a price reduction is possible for both values. Under the second-price rule, this is equivalent to Definition~\ref{def:privacy}, so we could adopt it as the definition without changing any of our second-price results.\footnote{\label{fn:price-reduction}Formally, let $v'=(v'_i,v_{-i})$ satisfy $F(v')=F(v)$, with common winner $w\ne i$. Require $z^P(v')=z^P(v)$ whenever, for each $\bar v_i\in\{v_i,v'_i\}$, there is an opponents' profile $\bar v_{-i}\in V^{n-1}$ such that
\[
w(\bar v_i,\bar v_{-i})=w
\qquad\text{and}\qquad
t_w(\bar v_i,\bar v_{-i})<t_w(v).
\]
Under the second-price rule, such a reduction is possible exactly when $\bar v_i<\ell(v)$. See also Remark~\ref{rem:sa-rulebased} of the Supplemental Appendix.} This formulation also provides a way to define privacy through the payment function for other pricing rules.

Lower-tail contextual privacy is feasible on its own. The descending implementation of Proposition~\ref{prop:grids}(ii) asks, from the top value downwards, whether a bidder's value is at least the current value and stops once the winner and the highest-priority price-setter have answered yes. No question has a threshold below $\ell(v)$, so the transcript is unchanged across the protected comparisons. The conflict arises when the seller must also have an incentive to follow the protocol.

\Needspace{12\baselineskip}
\subsection{Credibility under Private Observations}
To assess credibility, we consider a procedure $D$ that the seller
may use instead of the prescribed protocol $P$.
Such a procedure need not implement $F$.

A run $D$ is a deterministic procedure in which the seller chooses each query and its recipient based on earlier queries and answers, and eventually selects a feasible allocation and payments, with zero loser payments. Write $\Obs^D_i(v)$ for bidder $i$'s observation under $D$ at $v$. The run is \emph{safe} if, for every $v$ and $i$, some opponents' profile $\widehat v_{-i}\in V^{n-1}$ satisfies
\begin{equation}\label{eq:safe}
\Obs^D_i(v)=\Obs^P_i(v_i,\widehat v_{-i}).
\end{equation}
Call $\widehat v_{-i}$ an \emph{explanation} for $i$'s observation under $D$ at $v$: he could receive that observation if the seller followed $P$ at $(v_i,\widehat v_{-i})$.\footnote{We use \emph{honest} for the seller following $P$ and \emph{truthful} for bidders' responses.} Write $\Rev(D,v)$ for the seller's revenue (total payments received) under $D$ at $v$. In particular, $\Rev(P,v)=\ell(v)$.

\begin{definition}[Prior-free credibility]\label{def:credibility}
A protocol $P$ is \emph{prior-free credible} if no safe run satisfies $\Rev(D,v)\ge\Rev(P,v)$ at every $v$, with strict inequality somewhere.
\end{definition}

Following \citet{AkbarpourLi2020}, prior-free credibility is weaker than credibility under a specified full-support prior, which forbids any safe expected-revenue improvement. We allow withholding at zero revenue; requiring sale leaves the impossibility unchanged. Participants are fixed, without fictitious bidders or collusion.

On $V=\{0,\ldots,K\}$, Appendix~\ref{app:ascending} constructs an ascending implementation with announced histories that is prior-free credible under mandatory sale and protects the winner (Proposition~\ref{prop:scan}). Thus each requirement is feasible on its own with mandatory sale, yet they remain incompatible even with this restriction.

\section{Privacy Forces an Opportunity to Raise the Price}\label{sec:main}

\begin{theorem}[Second-price privacy and credibility]\label{thm:main}
For $n\ge 3$ and $m\ge 3$, no private query protocol implementing the second-price rule is both lower-tail contextually private and prior-free credible.
\end{theorem}

We prove the theorem in two steps, first characterizing the first informative query and then constructing a safe seller deviation.
\subsection{The First Informative Answer}

We call a query \emph{informative} if its response map $\rho_q$ is nonconstant on $V$. Before the first informative query, every profile generates the same initial sequence of queries and answers: no earlier answer distinguishes any values. These initial queries may be redundant and may involve several bidders.

\begin{lemma}[Privacy and the first question]\label{lem:first}
In a protocol implementing the second-price rule and satisfying lower-tail contextual privacy, the first informative query has two attainable answers: one shared by all types below $a_m$, and one identifying $a_m$.
\end{lemma}

\begin{proof}
Efficiency requires the winner to vary across profiles, so a first informative query exists. Let $i$ be its recipient. Fix two other bidders at $a_m$ and every remaining opponent at $a_1$. For any $v_i<a_m$, bidder $i$ is a low loser and the price is $a_m$. These values give the same winner and payments and are all strictly below the price. Lower-tail contextual privacy therefore makes their transcripts, and hence their answers, identical. No earlier answer has ruled out any value in $V$. Since the query is informative, $a_m$ gives a different answer, and no other attainable answers remain.
\end{proof}

\subsection{Replacing a Price Explanation}

For each bidder define
\[
W_i=\{u_{-i}\in V^{n-1}: w(a_m,u_{-i})=i\},\qquad M_i=\max_{u_{-i}\in W_i}t_i(a_m,u_{-i}).
\]
Thus, $W_i$ consists of the opponents' profiles against which bidder $i$ wins with value $a_m$, and $M_i$ is his largest payment across those profiles. The finite set $W_i$ is nonempty because value $a_m$ wins against opponents at $a_1$.

\Needspace{6\baselineskip}
\begin{lemma}[Replacing the winner's remaining queries and payment]\label{lem:replace}
Suppose one answer to the first informative query identifies bidder $i$'s value as $a_m$. There is a safe run that keeps the allocation under $P$, charges $M_i$ whenever $v_i=a_m$ and $w(v)=i$, and implements $F(v)$ otherwise. If bidder $i$ can win with value $a_m$ at different payments as opponents' values vary, then $P$ is not prior-free credible.
\end{lemma}

\begin{proof}
Fix $u^*_{-i}\in W_i$ attaining $M_i$. We construct a seller's run $D$. The seller follows $P$ through the first informative query and continues following $P$ unless the answer identifies $v_i=a_m$. With the identifying answer, she simulates the remaining steps of $P$: she sends each query to $j\ne i$ and receives his answer; for a query to $i$, she computes $\rho_q(a_m)$ and records the query without sending it yet. This determines the outcome $F(v)$.

If $i$ loses in the simulation, the seller sends him the recorded queries in their original order and implements $F(v)$. If he wins, she instead sends the remaining queries he would receive under $P$ at $(a_m,u^*_{-i})$ and allocates to him at payment $M_i$. His earlier queries also agree with this execution of $P$, since every profile shares the sequence through the first informative query. His answers agree throughout because $v_i=a_m$.

Each bidder $j\ne i$ receives the same queries and own outcome as under $P$ at $v$; when $i$ wins, these bidders lose and do not observe his payment. Thus $v_{-j}$ explains their observations. On the identifying branch, $i$ uses $v_{-i}$ if he loses and $u^*_{-i}$ if he wins; otherwise $D$ follows $P$. Delaying his queries changes the order of communication across bidders, which they do not observe. Hence $D$ is safe. Revenue is unchanged unless $v_i=a_m$ and $w(v)=i$; on that branch it rises weakly to $M_i$, and strictly whenever $t_i(v)<M_i$.
\end{proof}

Privacy supplies the early identifying event; private observations let different bidders rationalize their messages with different opponents' values.

\begin{proof}[Proof of Theorem~\ref{thm:main}]
Let $i$ receive the first informative query. Lemma~\ref{lem:first} gives an answer identifying $v_i=a_m$. At this value, he pays $a_1$ against opponents at $a_1$. If $i<n$, raising an opponent $j>i$ to $a_m$ gives him a winning payment of $a_m$. If $i=n$, an opponent at $a_{m-1}$ and all others at $a_1$ give payment $a_{m-1}$. Both exceed $a_1$ when $m\ge 3$, so Lemma~\ref{lem:replace} contradicts prior-free credibility.
\end{proof}

\section{The Size of the Available Price Increase}\label{sec:size}

Theorem~\ref{thm:main} identifies a bidder with the highest possible value, an event that can become unlikely on a fine grid. We bound available revenue increases across all lower-tail contextually private protocols and compute the expected gain for a descending construction. Fix $n\ge3$ and $0\le L<U$. For each $m\ge3$, let
\[
G_m=\{L,L+\Delta,\ldots,U\},\qquad \Delta=\frac{U-L}{m-1}.
\]

\begin{proposition}[Fine grids]\label{prop:grids}
\begin{enumerate}[label=(\roman*)]
\item For every $m\ge3$, every lower-tail contextually private protocol on $G_m$ admits a safe deviation that weakly increases revenue at every profile and increases it by at least $U-L-\Delta$ at some profile.
\item There are sequences $(P_m)_{m\ge3}$ of lower-tail contextually private descending implementations and $(D_m)_{m\ge3}$ of safe deviations from them such that $D_m$ weakly increases revenue at every profile and, under independent uniform draws from $G_m$,
\[
\lim_{m\to\infty}\mathbb{E}\bigl[\Rev(D_m,v)-\Rev(P_m,v)\bigr]=\frac{U-L}{n+1}>0.
\]
\end{enumerate}
\end{proposition}

\begin{proof}
For part (i), let $i$ receive the first informative query and apply Lemma~\ref{lem:replace} after its identifying answer. With value $U$, bidder $i$ pays $L$ against opponents at $L$, while his largest winning payment is at least $U-\Delta$. His payment can therefore increase by at least $U-L-\Delta$.

For part (ii), define $P_m$ by lowering a threshold from $U$ to $L$ in steps of $\Delta$. At each threshold, ask bidders in priority order whether their value is at least that threshold. The first \emph{yes} identifies the efficient winner $w=w(v)$; stop querying him and continue with the next bidder at the same threshold. The second \emph{yes} identifies the highest-priority price-setter; stop and charge $w$ the price $\ell(v)$. Questions at $L$ ensure that both affirmative answers occur. Every query asked at $v$ has threshold at least $\ell(v)$, so changing a low loser's value below $\ell(v)$ leaves the transcript and outcome unchanged.

Define $D_m$ by asking the same queries as $P_m$ but charging $w$ the amount $v_w-\Delta\mathbf1\{w=n\}$. The winner answers \emph{no} to each query at a threshold above $v_w$, then \emph{yes} at $v_w$, and receives no further queries. Conditional on his winning with value $v_w$, this sequence is independent of opponents' values. If $w<n$, set a lower-priority opponent's value to $v_w$ and every other opponent's value to $L$. Under $P_m$, bidder $w$ then wins, receives the same queries, and pays $v_w$. If $w=n$, winning implies $v_w>L$; setting bidder 1's value to $v_w-\Delta$ and all other opponents' values to $L$ gives the same winning observation with payment $v_w-\Delta$. Each loser retains his queries and own outcome under $P_m$ at $v$, so $D_m$ is safe.

The gain is $v_w-\ell(v)-\Delta\mathbf1\{w=n\}\ge0$.
The gap $v_w-\ell(v)$ equals $\Delta$ times the number of
grid boundaries separating the largest and second-largest values.
At the boundary between the $t$-th and $(t+1)$-st grid points,
each value lies below the boundary with probability $t/m$.
Under independent uniform draws, the probability that exactly
one value lies above it is therefore
$n(1-t/m)(t/m)^{n-1}$.
Summing these contributions gives
\[
\mathbb{E}[\Rev(D_m,v)-\Rev(P_m,v)]
=
\underbrace{
\Delta n\sum_{t=1}^{m-1}
\Bigl(1-\frac tm\Bigr)\Bigl(\frac tm\Bigr)^{n-1}
}_{\mathbb{E}[v_w-\ell(v)]}
-
\underbrace{\Delta\Pr(w=n)}_{\le\Delta\to0}.
\]
To obtain the limit, rewrite the first term as
\[
\Delta m\,n
\left[
\frac1m\sum_{t=1}^{m-1}
\Bigl(1-\frac tm\Bigr)\Bigl(\frac tm\Bigr)^{n-1}
\right].
\]
The bracketed expression is a Riemann sum for
$\int_0^1(1-x)x^{n-1}\,dx=1/[n(n+1)]$.
Since $\Delta m\to U-L$ and the correction term vanishes,
the expected revenue gain converges to $(U-L)/(n+1)$.
\end{proof}

The lower bound in part (i) approaches $U-L$ and applies to every lower-tail contextually private protocol. The expected-gain limit in part (ii) applies to the constructed descending protocols.

\section{Whose Privacy Can a Credible Auction Protect?}\label{sec:whose}

The impossibility also constrains the choice of bidders to protect. Call bidder $i$ \emph{globally protected} if, for every $v\in V^n$ and every $v'_i\in V$,
\[
\max\{v_i,v'_i\}<\ell(v)\quad\Longrightarrow\quad z^P(v'_i,v_{-i})=z^P(v).
\]
For a protocol implementing the second-price rule, define
\[
\begin{split}
\Gamma_L(P)=\{(i,v):\ &i\text{ is a low loser at }v,\ \exists v'_i\in V\text{ with }v'_i<\ell(v)\\
&\text{such that }z^P(v)\ne z^P(v'_i,v_{-i})\},
\end{split}
\]
the set of bidder-profile pairs violating the lower-tail contextual privacy condition~\eqref{eq:privacy}.

\begin{theorem}[The privacy cost of credibility]\label{thm:cost}
On $V=\{0,\ldots,K\}$ with $K\ge2$, any prior-free credible private query implementation of the second-price rule globally protects at most two bidders. Moreover,
\begin{equation}\label{eq:count}
|\Gamma_L(P)|\ge\binom n3 K(K-1)^{n-3}.
\end{equation}
\end{theorem}

Fix three bidders and give all others values at most $K-2$. Their first informative query must distinguish two values below $K$. Otherwise, identifying $K$ allows profitable replacement of that bidder's remaining queries and payment, as in Lemma~\ref{lem:replace}. With the other two members at $K$, those distinguished values lie below the price, violating lower-tail contextual privacy. Every triple therefore contains an unprotected bidder: at least $n-2$ bidders have violating profiles, possibly different ones. Appendix~\ref{app:count} counts these violations to obtain \eqref{eq:count}.%

Under mandatory sale, Proposition~\ref{prop:scan} shows that protecting the winner's surplus is compatible with credibility, while protecting every lower losing value is not.

\Needspace{16\baselineskip}
\section{Conclusion}\label{sec:conclusion}

The result distinguishes information needed to compute an outcome from information needed to sustain the seller's incentives. Protecting the strict lower tail forces the first informative query to identify a top type and exposes his payment to manipulation, even with arbitrary deterministic questions and redundant communication.

With mandatory sale, announced histories and a suitable question order yield a prior-free credible ascending implementation that protects the winner. If all bidders must receive the same terminal announcement of the price and a named price-setter, lower-tail privacy and credibility can coexist (Supplemental Appendix, Section~\ref{sa:record}; see also \citet[p.~13]{Woodward2020}). However, this announcement reveals the price-setting loser's value to all bidders.

The Supplemental Appendix (Section~\ref{sa:general}) extends the impossibility to efficient, incentive-compatible payment rules whose outcomes are unchanged by lower-tail changes. Whether the privacy requirement alone yields such an impossibility when lower losing values can affect payments remains open, as do continuous type spaces and private randomization.

\clearpage
\appendix

\section{A Credible Ascending Auction}\label{app:ascending}

On $V=\{0,\ldots,K\}$, a credible ascending implementation of the second-price rule requires a suitable question order.\footnote{Section~\ref{sa:clock} of the Supplemental Appendix describes the standard clock. Its winner pays his last accepted offer, which can exceed $\ell(v)$. Simply replacing that payment by $\ell(v)$ permits undetectable price increases.} Throughout this appendix, \emph{mandatory sale} means that every run allocates the object to a bidder. This restriction matters when the lowest value is zero: with withholding allowed, the protocol below admits a safe revenue improvement (Supplemental Appendix, Section~\ref{sa:clock}).

\paragraph{The protocol.} Prices rise from $1$ to $K$. All bidders are initially active, and a refusal makes a bidder inactive. Each query states its price $p$ and the complete history of earlier queries and answers, including those of other bidders, then asks whether the recipient's value is at least $p$. At each price:
\begin{enumerate}[label=(\roman*)]
\item \emph{First stage.} Ask active bidders in increasing priority order (largest index first) until one answers \emph{yes}. If only the highest-priority active bidder remains, stop before asking him; he wins and pays $p-1$.
\item \emph{Second stage.} Ask the remaining active bidders in decreasing priority order (smallest index first). For $p<K$, ask all of them even after an affirmative answer. If all answer \emph{no}, the first-stage acceptor wins and pays $p-1$. At $p=K$, stop at the first \emph{yes}; that bidder wins and pays $K$.
\item Otherwise, increase the price to $p+1$.
\end{enumerate}
The price rises only when the first-stage acceptor and a higher-priority bidder both accept, so at least two bidders are active at the start of every price. Each bidder's offered prices strictly increase, and the protocol ends by $K$. In an honest execution, a bidder knows the entire history through his last answer. These histories are privately stated seller messages, not a common record, and a deviating seller may falsify them.

\Needspace{7\baselineskip}
\begin{proposition}[Announced histories and question order]\label{prop:scan}
This protocol implements the second-price rule with payment $\ell(v)$ and is prior-free credible under mandatory sale. It satisfies contextual privacy for the winner: the full transcript is unchanged if his value changes to another value giving him the same winning outcome.
\end{proposition}

The next lemma relates the winner's payment to his queries and the history stated in them. Write $x$ for a payment and $p=0$ when the winner receives no query; this convention represents a zero payment without an actual offer at zero.

\begin{lemma}[Payment observability]\label{lem:observable}
Let $O$ be an honest observation in which bidder $i$ wins and pays $x$, and let $p$ be his last offered price, or zero if he receives no query. He accepts every offer, and one of the following holds:
\begin{enumerate}[label=(\alph*)]
\item $x=p$, and the history stated in $O$ records a refusal by every higher-priority bidder.
\item $x=p-1$, his last answer is the first acceptance at price $p$, and at least one higher-priority bidder has not yet refused in the history stated to $i$ with that query.
\end{enumerate}
\end{lemma}

\begin{proof}
A refusal makes a bidder inactive and unable to win, so $i$ accepts every query he receives, including his last one. Fix a profile generating $O$ and consider how the protocol awards him the object.

Suppose it stops in the first stage at price $q$, before querying $i$. He is then the sole remaining active bidder and pays $q-1$. If $q=1$, then $i=1$ and $p=x=0$, giving (a). Otherwise, his last query was at $q-1$, so $p=x=q-1$. He could not have been the first acceptor there: the protocol would then either have stopped at $q-1$ or retained a higher-priority bidder. In the latter case, $i$ would be queried before that rival at $q$, contrary to stopping before querying $i$. Hence his last query was in the second stage, after every higher-priority active bidder had refused. Those and earlier higher-priority refusals appear in his last query, giving (a).

Suppose instead that $i$ wins after all second-stage bidders refuse at price $q$. He accepted in the first stage at $q$, so $p=q$ and $x=q-1$. A higher-priority bidder was active when he was queried; otherwise the protocol would have stopped before asking him. That rival refuses after $i$'s last answer, so this refusal is absent from the history stated to $i$, giving (b).

Finally, suppose $i$ wins by accepting in the second stage at $K$. Then $p=x=K$. Every higher-priority bidder has already refused, and those refusals appear in his last query, giving (a).
\end{proof}

\begin{proof}[Proof of Proposition~\ref{prop:scan}]
\emph{Implementation.} If the first stage stops at $q$ with winner $i$, all rivals have values at most $q-1$, some with equality. Bidder $i$ has value at least $q-1$ and priority over every rival at that value, so $i=w(v)$ and $\ell(v)=q-1$. If the second stage stops after all remaining bidders refuse, the winner's value is at least $q$, while rivals' values are at most $q-1$, some with equality; thus the outcome is again correct. If $i$ accepts at $K$ in the second stage, he and the first acceptor have value $K$. All higher-priority bidders have refused, so $i=w(v)$ and $\ell(v)=K$.

\emph{Credibility.} Suppose a safe run $D$ with mandatory sale satisfies $\Rev(D,v)\ge\ell(v)$ at every $v$ and $\Rev(D,v^*)=r>\ell^*:=\ell(v^*)$ at some $v^*$. Write $w=w(v^*)$. Safety implies that a winner pays at most his value, a bidder who refuses cannot win, and each bidder's offered prices strictly increase. Determinism implies that identical answers generate identical subsequent queries and outcomes.

Only $w$ has value at least $r$, so he wins under $D$ at $v^*$ and accepts every offer. Let $p$ be his last offered price and $T$ the point when that query is sent. Lemma~\ref{lem:observable} gives $r\in\{p,p-1\}$. We consider two cases.

\emph{Case 1: $p\ge\ell^*+2$.} In case (a) of the lemma, $r=p\ge\ell^*+2$; in case (b), $r=p-1\ge\ell^*+1$. Let $v'=(p-1,v^*_{-w})$. Bidder $w$ remains the unique highest bidder and $\ell(v')=\ell^*$. All his offers before $T$ are at most $p-1$, so all answers before $T$ are unchanged. At $T$ he refuses price $p$ and cannot win. Mandatory sale requires another winner $j$, whose payment satisfies
\[
\ell^*\le\Rev(D,v')\le v^*_j\le\ell^*.
\]
Thus $v^*_j=\ell^*$ and $j$ pays $\ell^*$. He accepted every offer, all at prices at most $\ell^*$. Raise only $j$'s value in $v'$ to $\ell^*+1$. Every answer remains unchanged, so $D$ follows the same path and charges the same payment $\ell^*$. But $w$ and $j$ now both have values at least $\ell^*+1$, making the prescribed revenue at least $\ell^*+1$, a contradiction.

\emph{Case 2: $p=\ell^*+1$ and $r=p$.} Lemma~\ref{lem:observable}(a) implies that the history stated to $w$ at $T$ records a refusal by every higher-priority bidder. For each $i\ne w$, let $\lambda_i\le\ell^*$ be his largest accepted price before $T$ at $v^*$, or zero if he has accepted none. Set $\Lambda=\max_{i\ne w}\lambda_i$ and define $\bar v_w=\ell^*$ and $\bar v_i=\lambda_i$ for $i\ne w$. All answers before $T$ are unchanged: each $i\ne w$ accepted only prices at most $\lambda_i$ and refused only prices above $v^*_i\ge\lambda_i$, while $w$ received only offers at most $\ell^*$. Thus, also at $\bar v$, the same sequence reaches $T$ and $D$ sends $w$ the same query at price $p$. He refuses and cannot win, while $\ell(\bar v)=\Lambda$.

Mandatory sale gives a winner $j\ne w$. His payment $x$ satisfies
\[
\Lambda\le x\le\bar v_j=\lambda_j\le\Lambda,
\]
so $x=\lambda_j=\Lambda$. He accepted every offer. His last accepted price before $T$ is $\lambda_j$, with the zero convention if he received no query. Any later offer would exceed $\lambda_j$ and be refused, so he receives none. Lemma~\ref{lem:observable}(a) therefore implies that his observation before $T$ records a refusal by every higher-priority rival.

Now let $u$ agree with $v^*$ except that $u_w=u_j=K$. Both bidders accepted every offer before $T$, so all answers up to $T$ remain unchanged. The seller sends $w$ the same query at $T$, which he accepts. Both bidders' observations record refusals by all their higher-priority rivals. For a bidder with value $K$, such a record excludes every higher-priority rival from having value $K$ in any explanation. Every honest observation extending that record must therefore make him the winner. Safety at $u$ would require both $w$ and $j$ to win, contradicting feasibility.

Both cases lead to contradictions, so no safe run with mandatory sale weakly raises revenue everywhere and strictly somewhere.

\emph{Winner privacy.} Fix $v$ with winner $w=w(v)$, and let $p$ be his last offered price, or zero if he receives no query. In case (a) of Lemma~\ref{lem:observable}, $p=\ell(v)$. In case (b), $p=\ell(v)+1$ and a higher-priority rival has value $\ell(v)$, so $w$ can win only with a value at least $\ell(v)+1$. Every alternative winning value therefore gives the same answers to his queries. With opponents' values fixed, the full transcript and outcome are unchanged, as required by contextual privacy restricted to winners.
\end{proof}

\section{Proof of Theorem~\ref{thm:cost}}\label{app:count}
\begin{proof}
Fix $S\subseteq N$ with $|S|=3$ and $a\in\{0,\ldots,K-2\}^{N\setminus S}$. Consider all profiles $v$ with $v_{-S}=a$. As values in $S$ vary, the efficient winner changes, so some member of $S$ must receive an informative query. Let $i$ receive the first such query and let $h$ be the complete sequence of queries and answers strictly before it. Every profile with $v_{-S}=a$ generates $h$: earlier queries to members of $S$ have constant response maps, and outsiders' values are fixed.

Suppose this query gives the same answer for all values below $K$; its other answer identifies $v_i=K$. Construct a profile $\hat v$ with $\hat v_{-S}=a$ and $\hat v_i=K$. Give a lower-priority member of $S$ value $K$ if one exists; otherwise give a higher-priority member value $K-1$. Give the remaining member value zero. Bidder $i$ wins at $\hat v$, which reaches $h$, and pays $\hat p:=\ell(\hat v)\ge K-1$.

Use the same construction as Lemma~\ref{lem:replace}. Follow $P$ through $h$ and this query, continuing to follow $P$ on every other path or answer. If the answer identifies $v_i=K$, simulate the remaining queries of $P$, asking other bidders and computing $i$'s answers. If $i$ loses or would pay at least $\hat p$, send him the recorded queries and implement $F(v)$. Otherwise, send him his remaining queries from $\hat v$ and charge $\hat p$. Other bidders retain their queries and outcomes under $P$.

The explanatory profile $\hat v$ matches $i$'s observations through the identifying answer, even if different outsiders' values also generate $h$. Thus the run is safe and never reduces revenue. It strictly increases revenue when $i$ has value $K$, the other members of $S$ have value zero, and outsiders have values $a$: the prescribed price is at most $K-2<\hat p$. This contradicts credibility.

The query must therefore distinguish two values below $K$. For each $v_i\in\{0,\ldots,K-1\}$, set both other members of $S$ to $K$ and retain $a$ outside $S$. Some alternative $v'_i<K$ gives a different answer. Both values lie below $\ell(v)=K$ and give the same winner and payments, but their transcripts differ at this query. Hence each of these $K$ bidder-profile pairs belongs to $\Gamma_L(P)$.

Every triple thus contains a bidder who is not globally protected. Each choice of $(S,a)$ contributes $K$ distinct violations. There is no double counting: the profile has exactly two bidders at $K$; those bidders and the marked bidder $i$ determine $S$, and the profile then determines $a$ and $v_i$. The $\binom n3$ triples and $(K-1)^{n-3}$ choices of $a$ yield \eqref{eq:count}.
\end{proof}

\clearpage
\supplementsetup
\begin{center}\LARGE\bfseries Supplemental Appendix\end{center}\vspace{6pt}
\noindent This Supplemental Appendix collects results referred to in the main text: the standard ascending clock and its credibility under mandatory sale (Section~\ref{sa:clock}); the extension of the impossibility to general efficient, incentive-compatible payment rules, with a three-value example under mandatory sale (Section~\ref{sa:general}); and the common-record implementation discussed in the Conclusion (Section~\ref{sa:record}). Results stated here are numbered with the prefix S; unprefixed numbers refer to the main text. Throughout, $\ell(v)=\max_{j\ne w(v)}v_j$ denotes the highest losing value, $n\ge 3$, and lower indices have higher priority.

\section{The Standard Clock and the Clinching Price}\label{sa:clock}

This section describes the standard ascending clock referred to in Appendix~\ref{app:ascending}. On $V=\{0,\ldots,K\}$, write $w=w(v)$ and define the clinching price by
\[
R(v)=\min\{\bar v_w\in V:w(\bar v_w,v_{-w})=w\}
=\max\bigl(\{0\}\cup\{v_j+1:j<w\}\cup\{v_j:j>w\}\bigr).
\]
Thus $R(v)=\ell(v)+1$ when a price-setter has priority over the winner, and $R(v)=\ell(v)$ otherwise. The rule charging $R$ is normalized, efficient, DSIC, and ex post individually rational. The argument for Theorem~\ref{thm:main} also applies: a top bidder pays at most $a_2$ against bottom opponents and can pay $a_m$ against suitably chosen opponents.

The standard clock asks active bidders, in decreasing order of their indices, whether their values are at least $p$, for $p=1,\ldots,K$. A refusal permanently excludes a bidder. The clock stops when one bidder remains; if several remain after price $K$, the highest-priority bidder wins. The winner pays his last accepted offer, or zero if he received no query. This equals $R(v)$: he must accept one step above the highest losing value precisely when a price-setter has higher priority. More generally, consider ascending implementations of $R$ in which each bidder's thresholds strictly increase, a refusal excludes him from winning, and the winner pays his last accepted offer, with the same zero convention. As in Appendix~\ref{app:ascending}, mandatory sale requires every run to allocate the object.

\begin{proposition}[Binding offers]\label{prop:sa-ascending}
Under mandatory sale, every ascending implementation of $R$ with these properties is prior-free credible. It also leaves the full transcript unchanged when the winner's value changes to another value giving him the same winning outcome.
\end{proposition}

This conclusion concerns prior-free credibility under mandatory sale; it does not assert credibility under every prior.

\begin{proof}
In a safe run, any winner must have refused no price and must pay his largest accepted price, or zero if he receives no query. These restrictions hold in every honest own observation. In particular, his payment cannot exceed his value.

Suppose a safe run $D$ with mandatory sale weakly raises revenue everywhere and strictly raises it at $v$. Let $w=w(v)$, $\ell=\ell(v)$, and $R=R(v)$. A deviating payment $r>R$ must be made by $w$: every other bidder has value at most $\ell\le R$. Bidder $w$ therefore accepted the offer $r$.

Lower his value to $r-1$. Since $r-1\ge R$, he remains the efficient winner and the prescribed payment remains $R$. All answers before his offer $r$ are unchanged, because his own thresholds increase. He now refuses $r$, so safety prevents him from winning. Mandatory sale requires another winner, who can pay at most $\ell$. Weak revenue improvement therefore requires $\ell=R$ and a price-setter $j$ who pays $\ell$ and has accepted no offer above $\ell$ and refused none.

Raise this bidder's value from $\ell$ to $\ell+1$. The run of $D$ is unchanged: every question it asked $j$ had threshold at most $\ell$, and all other values stay fixed. It still collects $\ell$. The honest critical price, however, is at least $\ell+1$. If the modified value $r-1$ of $w$ is at least $\ell+1$, the two bidders $w,j$ have values at least $\ell+1$. If $r-1=\ell$, then $R=\ell$ implies that $w$ has priority over $j$, so $j$ must bid $\ell+1$ to beat $w$. Either way $D$ now earns less than the rule, a contradiction. The changed value belongs to the grid because $\ell=R<r\le K$.

For winner privacy, fix an honest winner with payment $R$. His largest accepted threshold is $R$ and he has refused none. Any alternative winning type is at least $R$ and gives the same answers to all his questions. Keeping opponents fixed therefore preserves the entire path and outcome. Conversely, every type at least $R$ wins and has the same critical payment.
\end{proof}

\begin{remark}[Why the order matters]\label{rem:sa-order}
Charging $\ell(v)$ instead of $R(v)$ in the standard clock is not credible even under mandatory sale. A winner may have accepted his final offer while a higher-priority rival remains unresolved in his observation. For example, on $V=\{0,1\}$ with three bidders at $(0,1,1)$, bidder 2 accepts 1 before bidder 1 refuses, then wins and pays 1. The seller can give bidder 2 the same query when bidder 3 has actually refused, and charge 1 if bidder 1 also refuses. The payment-observability lemma (Lemma~\ref{lem:observable}) fails for this order. Appendix~\ref{app:ascending} instead queries every higher-priority bidder before the winner whenever a lower-priority bidder has accepted. Case~1 in the proof of Proposition~\ref{prop:scan} then rules out charging a first-stage acceptor his last accepted offer.
\end{remark}

\paragraph{Why mandatory sale matters at zero.}
Allowing withholding invalidates the credibility claim for the zero-based protocol in Appendix~\ref{app:ascending}. To see this, take three bidders and $K\ge2$. Follow that protocol unless bidders 3 and 2 both refuse price 1. On this branch both have value zero. Instead of giving bidder 1 the object at price zero, ask him at price 1 with a stated history in which bidder 3 accepted 1. If bidder 1 accepts, allocate to him at price 1; otherwise withhold the object. Bidders 2 and 3 retain observations explained by a positive value of bidder 1. Bidder 1's observation, including either outcome, is explained by the honest profile $(v_1,0,1)$. The run is therefore safe. Its revenue equals the prescribed revenue on other branches, rises from zero to 1 when $v_1\ge1$ on this branch, and remains zero when $v_1=0$. Thus mandatory sale is essential for this construction on a grid containing zero. The main impossibility results do not require it.

\section{General Incentive-Compatible Payments}\label{sa:general}

On a discrete type space, normalized, efficient, DSIC, ex post individually rational rules need not have unique payments. This section extends Theorem~\ref{thm:main} to such rules when their outcomes ignore lower-tail changes. Throughout, $F=(w,t_1,\ldots,t_n)$ is such a rule with the efficient allocation $w$ of Section~\ref{sec:model}, and a loser $i$ is a low loser at $v$ if $v_i<\ell(v)$.

For these rules the protected range is still defined by $\ell(v)=\max_{j\ne w(v)}v_j$, the highest losing value, which depends only on the allocation. Efficiency and DSIC place the winner's payment between the grid neighbors of $\ell(v)$ (see \eqref{eq:sa-interval} below), so a losing value strictly below $\ell(v)$ can affect the outcome only through the rule's own use of lower-tail information.\footnote{Defining the lower tail by the winner's payment $t_{w(v)}(v)$ instead of $\ell(v)$ would narrow the class: together with DSIC, outcome invariance below the payment forces $t_{w(v)}(v)\le\ell(v)$ whenever $\ell(v)>a_1$, because the price-setters with priority over the winner can then be lowered to $a_1$ one at a time without changing the outcome, and once none of them remains the winner's smallest winning type is at most $\ell(v)$. This excludes the clinching price of Section~\ref{sa:clock}, under which such a price-setter lies strictly below the price. It would also weaken Lemma~\ref{lem:first}, since the payment at the profile used there could be $a_{m-1}$.} We use an outcome-conditioned extension of Definition~\ref{def:privacy}: for every $v\in V^n$, loser $i\ne w(v)$, and $v'_i\in V$, require
\begin{equation}\label{eq:sa-privacy-general}
F(v'_i,v_{-i})=F(v)\ \text{ and }\ \max\{v_i,v'_i\}<\ell(v)\quad\Longrightarrow\quad z^P(v'_i,v_{-i})=z^P(v).
\end{equation}
To ensure that all comparisons in this range leave the outcome unchanged, impose \emph{lower-tail outcome invariance}:
\begin{equation}\label{eq:sa-invariance}
F(v)=F(v'_i,v_{-i})\quad\text{for every low loser $i$ and every $v'_i<\ell(v)$.}
\end{equation}
Numerical second pricing satisfies \eqref{eq:sa-invariance} automatically; for general rules it is an additional assumption. It makes the outcome-equality qualification in \eqref{eq:sa-privacy-general} redundant, recovering Definition~\ref{def:privacy}. For any efficient rule satisfying \eqref{eq:sa-invariance}, the first-query proof of Lemma~\ref{lem:first} therefore applies.

\Needspace{7\baselineskip}
\begin{theorem}[General incentive-compatible payments]\label{thm:sa-general}
Let $m\ge 4$. No private query protocol implementing a normalized, efficient, DSIC, ex post individually rational rule satisfying lower-tail outcome invariance \eqref{eq:sa-invariance} can both satisfy the privacy condition \eqref{eq:sa-privacy-general} and be prior-free credible. On the equally spaced grid $G_m$ of Proposition~\ref{prop:grids}, there is a safe pointwise revenue improvement of at least $U-L-2\Delta$ at some profile.
\end{theorem}

\begin{proof}
Fix a bidder and his opponents. DSIC makes payments constant across winning reports: comparing any two winning reports in both directions gives equality. If $b$ is the smallest winning type and $c$ the largest losing type, normalized DSIC implies
\begin{equation}\label{eq:sa-interval}
c\le t_i\le b.
\end{equation}
The lower bound prevents type $c$ from profiting by winning; the upper bound prevents type $b$ from profiting by losing. If all reports win, ex post individual rationality gives $t_i\le a_1$.

Give bidder $i$ value $a_m$. Against bottom opponents, his payment is at most $a_1$ if $i=1$ and at most $a_2$ otherwise. If $i=1$, a lower-priority opponent at $a_m$ makes every report at most $a_{m-1}$ lose, so his winning payment is at least $a_{m-1}$. If $i>1$, a higher-priority opponent at $a_{m-1}$ gives the same lower bound. Since $a_{m-1}>a_2$ when $m\ge 4$, every bidder has nonconstant top-type winning payments. The first-query argument of Lemma~\ref{lem:first}, together with Lemma~\ref{lem:replace}, gives a safe pointwise revenue improvement. On the equally spaced grid, the difference between these high- and low-opponent bounds is at least $(U-\Delta)-(L+\Delta)=U-L-2\Delta$.
\end{proof}

Under mandatory sale, the four-value bound cannot be lowered, as the following example shows.

\begin{example}[A three-value example]\label{ex:sa-three}
On $V=\{0,1,2\}$, allocate efficiently. A winning bidder $i>1$ pays 1. A winning bidder 1 pays 1 if an opponent has value 2, and zero otherwise. Losers pay zero. This normalized, efficient, DSIC, individually rational rule satisfies lower-tail outcome invariance and has a lower-tail contextually private implementation that is prior-free credible under mandatory sale.
\end{example}

The rule is lower-tail outcome invariant: replacing a low loser's value by another value below $\ell(v)$ leaves the winner and every payment unchanged. Implement it with these private queries, where a query at $p$ asks whether the recipient's value is at least $p$:
\begin{enumerate}[label=\arabic*.]
\item Ask bidders $n,n-1,\ldots,2$ at 2, including those following an affirmative answer.
\item If any answer is affirmative, ask bidder 1 at 2. Allocate to the highest-priority bidder answering affirmatively at 2, and charge 1.
\item If all initial answers are negative, ask bidder 1 at 1. An affirmative answer gives him the object at payment zero.
\item Otherwise, ask bidders $2,3,\ldots,n$ at 1, stopping at the first affirmative answer. That bidder wins and pays 1. If every answer is negative, bidder 1 wins and pays zero.
\end{enumerate}

\begin{proof}[Verification of Example~\ref{ex:sa-three}]
If any value is 2, the protocol identifies the highest-priority such bidder and charges 1, except when bidder 1 alone has value 2: he wins in the third step at payment zero. Otherwise it selects the highest-priority value-1 bidder, or bidder 1 if all values are zero. Thus it implements the rule.

Fix opponents' reports. Normalized DSIC requires a constant payment across winning reports, between the largest losing report and the smallest winning report whenever both exist. For bidder 1, opponents' maxima 0, 1, and 2 give smallest winning reports 0, 1, and 2; his prescribed payments 0, 0, and 1 satisfy these conditions, including individual rationality when all reports win. For $i>1$ with a winning report, the smallest winning report is 1 or 2, and payment 1 belongs to both intervals $[0,1]$ and $[1,2]$. A bidder without a winning report always pays zero. Thus the rule is DSIC and individually rational.

For lower-tail privacy, when $\ell(v)\le1$, the strict lower tail contains at most one value. When $\ell(v)=2$, at least two bidders have value 2. A low loser's values 0 and 1 both give negative answers at 2; some other bidder among $2,\ldots,n$ has value 2, so the protocol takes the second step and never asks that low loser at 1. The protected transcripts coincide.

For credibility, bidder 1 receives exactly one honest query. At price 2, an affirmative answer entails winning at payment 1, and a negative answer entails losing. At price 1, an affirmative answer entails winning at payment zero; a negative answer permits either losing or winning at zero. Every other bidder who accepts price 1 necessarily wins and pays 1; these bidders never win with value zero. Bidder 1 always wins with value 2. Safety imposes all these restrictions on every completed deviating observation.

No bidder $j>1$ can accept price 1 in a safe run before bidder 1 is contacted. Otherwise, changing the uncontacted bidder's value to 2 preserves that prefix, and safety requires both $j$ and bidder 1 to win, contradicting feasibility.

Suppose a safe run with mandatory sale weakly increases revenue everywhere and strictly increases it at $v$. Safe winning payments are at most 1, so prescribed revenue at $v$ is zero: bidder 1 wins and opponents have values at most 1. If $v_1=0$, all values are zero, and no safe payment can raise revenue. Thus $v_1\ge1$. Raising revenue requires bidder 1 to pay 1 or another bidder to win, both of which require bidder 1's sole query to be at 2.

Before that query, all answers at 2 are negative, and the preceding no-contact argument excludes affirmative answers at 1. The same prefix therefore occurs at $\bar v=(1,0,\ldots,0)$. Bidder 1 refuses at 2 and must lose. Every other bidder has value zero and cannot safely win, contradicting mandatory sale.
\end{proof}

The example uses the intersection of the DSIC intervals $[0,1]$ and $[1,2]$. Adding a positive constant to every value and every winning payment yields the corresponding example when withholding is allowed: the final contradiction above becomes a strict revenue loss from withholding. Hence the bound $m\ge4$ is sharp for the general class of nonnegative finite value sets in Theorem~\ref{thm:sa-general}. On the zero-based grid itself, the stated example uses mandatory sale. None of these claims changes the second-price impossibility with three values in Theorem~\ref{thm:main}.

Outcome invariance has substantive content. DSIC allows the winner's payment, holding opponents fixed, to be any constant in \eqref{eq:sa-interval}, and different opponents' profiles may select different points of this interval. A rule can therefore assign distinct prices to different lower losing reports without changing the reporting bidder's own allocation or payment. Whenever the prices differ, $F$ changes and the premise of \eqref{eq:sa-privacy-general} fails: the outcome-conditioned criterion no longer forbids distinguishing the corresponding reports, and an injective payment code can make it vacuous. In the terminology of \citet{HauptHitzig2026}, \eqref{eq:sa-invariance} together with \eqref{eq:sa-privacy-general} is their individual contextual-privacy criterion on the strict lower tail, that is, contextual privacy plus lower-tail nonbossiness (their Proposition 8): a low loser's own outcome is always no object and zero payment, and his report may not change anyone else's payment.

Credibility imposes a further necessary restriction: whenever the same bidder wins with the same complete own query--answer sequence, his payment must be the same, for otherwise the seller can follow the protocol honestly and replace only the final payment by the largest honest payment compatible with his observation; this is the finite deterministic version of the extensive pay-as-bid principle of \citet[Theorem 5]{AkbarpourLi2020}. Thus, if a lower losing value changes the winning payment, a credible implementation must also change that winner's own messages. Whether the outcome-conditioned requirement \eqref{eq:sa-privacy-general} alone implies the general impossibility without outcome invariance is not resolved here.

\Needspace{12\baselineskip}
\begin{remark}[Rule-based lower tails]\label{rem:sa-rulebased}
The price-reduction test in footnote~\ref{fn:price-reduction} requires, for any $v'=(v'_i,v_{-i})$ with $F(v')=F(v)$ and common winner $w\ne i$, that $z^P(v')=z^P(v)$ whenever, for each $\bar v_i\in\{v_i,v'_i\}$, there is an opponents' profile $\bar v_{-i}\in V^{n-1}$ such that
\[
w(\bar v_i,\bar v_{-i})=w
\qquad\text{and}\qquad
t_w(\bar v_i,\bar v_{-i})<t_w(v).
\]
For numerical second pricing this test is equivalent to Definition~\ref{def:privacy}: the price $\ell(v)$ is at least every losing value, so a loser's value can be kept while the same winner pays less exactly when it lies strictly below $\ell(v)$.
\end{remark}

\section{Implementations with a Common Terminal Record}\label{sa:record}

The Conclusion observes that a common announcement naming the price-setter changes what each bidder can check against his own value. This section makes the statement precise for numerical second pricing. The extension adds to a private query protocol a \emph{terminal record}: after the last answer, the seller announces one message identically to all bidders. A bidder's observation consists of his own query--answer sequence, the record, and his own allocation and payment. A run with a record is a deterministic procedure that chooses each private query and its recipient from earlier queries and answers, then announces one record to all bidders and assigns a feasible allocation with zero loser payments; safety uses the observation-matching condition \eqref{eq:safe}, with observations including the record. The restriction is that the record is common: the seller cannot send different bidders different records. A protocol with a record is \emph{credible under every prior} if, for every full-support prior, no safe run has strictly higher expected revenue. Since full-support priors can concentrate on any profile, this holds if and only if no safe run earns more than the protocol at any profile; it implies prior-free credibility.

Fix the descending implementation $P_m$ of Proposition~\ref{prop:grids}: the threshold decreases from the largest grid value, bidders are asked in priority order at each value whether their value is at least the current value, the first yes identifies the winner $w$, who is then removed, and the second yes, at value $\ell(v)$, identifies the highest-priority price-setter $d$, after which the protocol stops and charges $\ell(v)$; questions at the lowest value are retained so that both yes answers always occur. Let $P_m^{\circ}$ be $P_m$ followed by the record $(d,\ell(v))$.

\begin{proposition}[A named price-setter]\label{prop:sa-record}
$P_m^{\circ}$ implements numerical second pricing and satisfies lower-tail contextual privacy. Every safe run with a record earns at most $\ell(v)$ at every profile $v$. Hence $P_m^{\circ}$ is credible under every full-support prior, and in particular prior-free credible.
\end{proposition}

\begin{proof}
Implementation and privacy are as in Proposition~\ref{prop:grids}: every question has threshold at least $\ell(v)$, and the record depends only on the price-setter and the price, so changing a low loser's value below $\ell(v)$ changes neither the transcript, nor the record, nor the outcome.

Honest observations have three features. First, the winner's payment equals the price in the record, and the record never names the winner. Second, the named bidder was asked whether his value is at least the recorded price and answered yes, so his value is at least that price. Third, a bidder's own answers are truthful, so a winner's payment is at most his value.

Let $D$ be a safe run with a record, and let $v$ be a profile at which $D$ sells, to bidder $w'$ at payment $t'$, with record $(d',p')$. Safety for $w'$ requires $t'=p'$, $d'\ne w'$, and $p'\le v_{w'}$. Safety for $d'$ requires that $d'$ answered yes to the question at $p'$, so $v_{d'}\ge p'$. Suppose $p'>\ell(v)$. Every bidder other than $w(v)$ has value at most $\ell(v)$, so $v_{d'}\ge p'$ forces $d'=w(v)$, and $v_{w'}\ge p'$ forces $w'=w(v)$, contradicting $d'\ne w'$. Hence $p'\le\ell(v)$, and $D$ earns $p'\le\ell(v)$ at $v$. If $D$ does not sell at $v$, it earns $0$. The remaining claims follow.
\end{proof}

Separately delivered private messages do not provide this restriction: if the seller could tell $w'$ that $d'$ set the price and tell $d'$ nothing, the argument above fails, and Theorem~\ref{thm:main} applies.

\clearpage
{\small\setstretch{1.0}\setlength{\bibsep}{2pt}

}

\end{document}